\documentclass[11pt]{article}
\usepackage{fullpage}

\usepackage[colorlinks=true,linkcolor=blue]{hyperref}

\usepackage{verbatim}

\usepackage{amsthm}

\usepackage{babel}
\usepackage{amsmath,amsfonts,amssymb}
\usepackage{soul,thmtools}
\usepackage[capitalise]{cleveref}
\usepackage{complexity}
\usepackage{graphicx}
\usepackage[skip=10pt plus1pt, indent=10pt]{parskip}
\usepackage[    protrusion=true,
expansion=true,
final,
babel
]{microtype}
\usepackage{palatino}
\usepackage{tabularx}
\usepackage[textsize=small]{todonotes}
\usepackage[linesnumbered,ruled,vlined]{algorithm2e}
\usepackage{tikz}
\usetikzlibrary{fit, backgrounds, positioning}
\usetikzlibrary{shapes.geometric, arrows.meta,decorations.pathmorphing, decorations.markings}

\newcommand{\REACH}{\mathsf{Reach}}
\newcommand{\REACHUL}{\mathsf{ReachUL}}
\newcommand{\RUL}{\REACHUL}
\newcommand{\Rec}{\mathsf{Rec}}

\newcommand{\GREACH}{\mathcal{G}\text{-}\REACH}

\newcommand{\SPREACH}{\mathsf{SimplePath\text{-}\REACH}}

\newcommand{\GREC}{\mathcal{G}\text{-}\Rec}
\newcommand{\ISRBACK}{\textsf{IsRightChildBackEdge}}
\newcommand{\ISLBACK}{\textsf{IsLeftChildBackEdge}}
\newcommand{\NEXT}{\textsf{Next}}

\newcommand{\PATHUNCLE}{\textsf{PathUncle}}
\newcommand{\MAIN}{\textsf{Main}}

\newcommand{\etal}{\textit{et al.}}

\newtheorem{theorem}{Theorem}
\newtheorem{lemma}[theorem]{Lemma}
\newtheorem{proposition}[theorem]{Proposition}

\theoremstyle{definition}
\newtheorem{definition}[theorem]{Definition}

\theoremstyle{remark}
\newtheorem{remark}[theorem]{Remark}

\usepackage{url}

\crefname{algocf}{Alg.}{Algs.}
\Crefname{algocf}{Algorithm}{Algorithms}

\crefname{linecf}{Line.}{Lines.}
\Crefname{linecf}{Line}{Lines}

\begin{document}
\title{Space Complexity of Reachability in Simple Path Graphs}
\author{Krishnamoorthy Dinesh$^*$ \and Chandana Sasidharan\footnote{Department of Computer Science and Engineering, Indian Institute of Technology Palakkad, Kerala, India. \phantom{xxxx} Email: \texttt{\{kdinesh,112303001\}@iitpkd.ac.in}}}

\maketitle              

\begin{abstract}
	One of the central questions in space complexity is whether nondeterministic logspace computations ($\NL$) can be simulated in unambiguous logspace ($\UL$). A stronger notion, \emph{reach unambiguity} ($\RUL \subseteq \UL \cap \co\UL$), requires every configuration reachable from the start to have exactly one computation path~\cite{BGL91}. It is known that directed graph reachability is $\NL$-complete, planar graph reachability is in $\UL$, and undirected graph reachability is in deterministic logspace ($\L$). In this work, we study the space complexity of reachability problem for restricted directed graph families and show the following.
	\begin{enumerate}
		\item For reach unambiguous graphs (graphs with at most one path from start vertex to every vertex), Lange~\cite{LAN97} showed that reachability is in $\RUL$. As our main result, we show that for simple path graphs (introduced in~\cite{KKR08}, which contains reach unambiguous graphs), where each vertex has at most one \textit{simple} path from the start, the reachability problem lies in $\UL \cap \co\UL$. The key difficulty lies in recognizing whether the input graph is a simple path graph or not.
		
		\item Our first result can also be equivalently stated as follows: the recognition problem for simple path graphs is in $\UL \cap \co\UL$ if and only if the reachability problem restricted to simple path graphs is also in $\UL \cap \co\UL$. Inspired by this, we investigate the complexity of graph recognition versus graph reachability for other directed graph classes. Observe that for any graph class, solving reachability (for the class) also solves the recognition problem for that class. We show that for reach unambiguous graphs, solving recognition is as hard as solving reachability (making both of them $\RUL$-complete).   
	\end{enumerate}

\end{abstract}

\newpage
\section{Introduction}\label{sec:intro}
One of the central problems in space complexity asks if nondeterministic computation (captured by the class $\NL$ consisting of languages accepted by machines using logspace with potentially multiple accepting paths on any run) can be made unambiguous (captured by the class $\UL$ consisting of languages accepted by machine using logspace with at most one accepting path on any run)~\cite{RA2000}. A stricter form of unambiguity, namely, \textit{reach unambiguity} (captured by the class $\RUL \subseteq \UL\cap \co\UL$) introduced in \cite{BGL91}  requires that for every configuration reachable (which can be accepting, rejecting or aborting) from the start configuration, there is exactly one path in the configuration graph of the machine\footnote{See \cref{sec:prelims} for a formal definition.}. Note that there are two existing definitions for $\RUL$. Buntrock \etal~\cite{BGL91} says that the configuration graph of the run of a reach unambiguous machine could have multiple accepting configurations, while Lange's \cite{LAN97} definition has a fixed unique accepting configuration. Both definitions yield the same class \cite{BGL91}. In this work, we use Lange's definition.

\cref{fig:allcomplexityclasses} gives a visual representation of how the configuration graph of machines associated with each of the above complexity classes looks like.

\begin{figure}[h]
	\centering
	\includegraphics[width=0.87\linewidth]{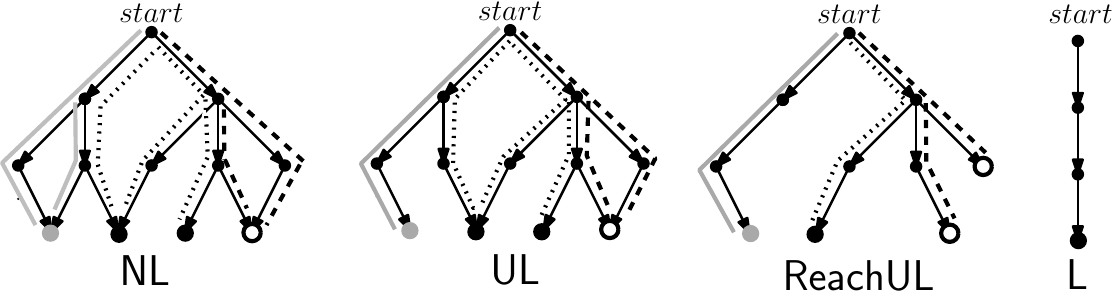}
	\caption{The gray vertices denote the accepting configuration, solid black vertices denote reject configuration, and black outlined vertices denote abort configuration. The gray paths are accepting paths, dotted paths are rejecting and dashed paths are aborting. The path constraints are as follows: $\NL$ allows an arbitrary number of paths, $\UL$ requires at most one accepting path, $\RUL$ requires at most one accept, reject, and abort paths (to each vertex), $\L$ has exactly one computational path ending in either accept or reject configuration.}
	\label{fig:allcomplexityclasses}
\end{figure}

A natural way to get a better understanding of these classes is by studying problems complete for them (under many-one logspace reductions). While the general problem of directed graph reachability is complete for $\NL$, we do not know of any natural complete problem for $\UL$. 

Surprisingly, for the stricter notion of reach unambiguity, Lange \cite{LAN97} showed that reachability for a certain restricted family of graph is complete for $\RUL$ (under logspace reductions). Of particular interest in this work is a special family of directed graphs called reach unambiguous graphs (where there is at most one path from start vertex to every other vertex) for which the reachability problem is in $\RUL$. This family of graphs is of interest as it is the only family of directed graphs for which reachability for $n$ vertex graphs can be checked deterministically in $O(\frac{\log^2 n}{\log\log n})$ space, beating what is achievable with Savitch's theorem (for more details, see~\cite{AL99,E23}). For connections to logspace algorithms in the quantum setting, see~\cite{AE24}). Also, showing that a problem is in $\REACHUL$ implies that it can be solved in in simultaneous polynomial time and poly-logarithmic space, which places it in the class $\SC^2$~\cite{BGL91,CookS79}. It is an open problem to design such algorithms solving reachability for general directed graphs~\cite{Wig}.

Kannan \etal~\cite{KKR08} considered a strict generalization of this family called \textit{simple path} graphs, where there is at most one \textit{simple} path from the start vertex to any other reachable vertex\footnote{This is indeed a strictly larger family, as the vertices reachable from the start vertex must form a tree in reach unambiguous graphs while a simple path graph is allowed to contain directed cycles. 
}.They gave a simultaneous sub-linear space and polynomial time algorithms for solving graph reachability restricted to this class. Towards improving this result, if we could show that the reachability for simple path graphs is solvable reach unambiguously in logspace, then it gives a polynomial time and poly-logarithmic space algorithm (as $\REACHUL \subseteq \SC^2$). 

Towards this, as our first contribution, we show that the reachability problem for simple path graphs (denoted $\SPREACH$) is in $\UL \cap \co\UL$ (\cref{Thm:SPReachULCoUL}). 
Note that all the graphs considered in this work are directed.

\begin{theorem}[Main]\label{Thm:SPReachULCoUL}
	Given a graph, checking if it is a simple path graph, is in $\UL \cap \co\UL$. In particular, this means that the reachability problem restricted to this class $\SPREACH \in \UL \cap \co\UL$.
\end{theorem}

Observe that if a graph $G$ with $s$ as a start vertex is guaranteed to be a simple path graph, then there is at most one simple path to any other vertex reachable from $s$. Hence, such a path must also be the shortest path, which is unique. A result of Reinhardt and Allender~\cite{RA2000} shows that given a graph $G$ and two vertices $s$ and $t$, if there is a logspace computable weight assignment to the edges of a graph that guarantees that the minimum weight path from $s$ to $t$ is unique, checking reachability from $s$ to $t$ in $G$ is in $\UL \cap \co\UL$. Hence, assigning unit weights to each edge (which is indeed logspace computable), immediately shows that the reachability problem $\SPREACH$ is in $\UL\cap\co\UL$. Hence, the difficult part is in recognizing whether the given graph is indeed a simple path graph\footnote{We thank the anonymous reviewers for pointing this out to us.}. In such graphs, the vertices that are \textit{not} reachable from the start vertex could have arbitrary structure (see \cref{fig:simplepathgrah}). This also adds to the difficulty of the recognition problem for simple path graphs.

Due to this, it is not clear if one could design a (logspace computable) weighting scheme to check if a given graph is a simple path graph (via an application of Reinhardt and Allender's result). Instead, we build on the algorithm of Lange~\cite{LAN97} for reach unambiguous graphs (which can be viewed as a depth bounded depth first search) and extend the same to simple path graphs. 

Motivated by the relation between the recognition and reachability for simple path graphs, as our second contribution, we explore this relation for other graph classes also. In general, the complexity of reachability restricted to a specific family of graphs seems to give a lot of information on the power and limitations of space bounded complexity classes. A celebrated result due to Reingold~\cite{Rein08} showed that if the graphs are undirected, then the reachability problem is solvable in logspace. On the other hand, the problem stays complete for $\NL$ if the graphs are restricted to be acyclic. Interestingly, the reachability problem is solvable in $\UL$ if the graph family is planar~\cite{BTV09} (or even of higher genus~\cite{GST19}).

In this context, it is natural to ask about the complexity of the reachability problem for a special family of graphs $\mathcal{G}$, called the $\GREACH$ problem. For such a problem, it is also necessary to check if the graph belongs to the family of graphs $\mathcal{G}$, called as the $\GREC$ problem. Note that, in general, the reachability problem for a family of graphs can be harder than the recognition problem. For instance, if $\mathcal{G}$ is the family of directed bipartite graphs $G=(S,T,E)$, where the edges may go from $S$ to $T$, $T$ to $S$ or both, checking if the underlying undirected graph $G$ is bipartite is in logspace~\cite{ACR00}, while reachability restricted to bipartite graphs is complete for $\NL$. We ask if the same holds true for complexity classes below $\NL$.

When $\mathcal{G}$ is reach unambiguous graphs, we show that the recognition problem is as hard as the reachability problem (\cref{Thm:Mangrove_reach}), thereby exhibiting a class below $\NL$ for which the previous statement is not true. Note that, since $\GREACH$ for reach unambiguous graphs is in $\RUL$ \cite{LAN97}, and as the recognition problem reduces to it, the recognition problem for reach unambiguous graphs is also solvable in $\RUL$. Hence the recognition and reachability problems reduces to each other for reach unambiguous graphs.

\begin{proposition} \label{Thm:Mangrove_reach}
	Let $\mathcal{G}$ be reach unambiguous graphs, let $M$ be a reach unambiguous machine by which $\GREC$ is in $\RUL$. Then, there exists a reach unambiguous machine which makes exactly one oracle query to $L(M)$, solving the $\GREACH$ problem.
\end{proposition}

Towards proving the above, note that a $\RUL$ algorithm with a $\RUL$ oracle access can be simulated in $\RUL$. Such a result is known when the base machine is in deterministic logspace~\cite{BGL91}. We state this as a proposition.

\begin{restatable}{proposition}{RULtoRUL} \label{Thm:RUL^RUL}
	$\RUL^{\RUL} = \RUL$
\end{restatable}

Indeed, $\RUL \subseteq \RUL^\RUL$. To see the other direction, any machine $N$ with an oracle access to $L'$ in $\RUL$, where $N$ accepts a language in $\RUL$ can be simulated reach unambiguously. As $\RUL = \co\RUL$ \cite{BGL91}, the answers to oracle queries in $L'$ can be found reach unambiguously.

\paragraph{Our techniques and overview of proofs: }We describe the main ideas that goes in proving \cref{Thm:SPReachULCoUL}. 
We start by describing Lange's procedure and highlight the challenges involved in solving reachability for simple path graphs.

Lange's procedure can be viewed as a depth bounded depth first search. More precisely, the procedure proceeds in rounds (starting with $d=1$), where each round starts by performing a pre-order traversal of the subgraph of the original graph consisting of all the vertices that are reachable from the start vertex up to a fixed depth $d$. To facilitate the pre-order traversal, the procedure takes advantage of the graph structure and uses nondeterminism to correctly identify the next vertex in the pre-order traversal. If the nondeterministic choices are correct, the accepts and rejects will be unambiguous. The non-trivial part of Lange's algorithm is that if the nondeterministic choices are incorrect, the procedure aborts and does so unambiguously, making the overall algorithm reach unambiguous.

For each value of $d$, the procedure verifies that the above subgraph (vertices reachable from $s$ at distance at most $d$) is indeed reach unambiguous. The run of this procedure eventually defines a DFS tree rooted at $s$ consisting of all the vertices reachable from $s$. Note that this tree is never stored explicitly.

We now discuss the challenges involved in adapting Lange's procedure to simple path graphs. Simple path graphs are those graphs for which any DFS tree starting from $s$ does not have a forward or cross edges (See \cref{lemma:char-of-simple-path}). It can be argued that Lange's procedure \cite{LAN97} detects if the input graph has a forward, back edge or cross edge and rejects in all the three cases unambiguously. In particular, this means that the procedure will reject simple path graphs. 

The first challenge is to make sure that the procedure does not reject on a back edge. This requires us to unambiguously identify if an edge is a back edge or not. We achieve this using the procedures $\ISLBACK$ (\cref{alg:isleftchildbackedge}) and $\ISRBACK$, which takes in a vertex $y$ and height $h$ of $y$ in the DFS tree and checks if its left or right child forms a back edge unambiguously. However, it need not abort unambiguously. Note that Lange's algorithm \cite{LAN97} does not need to perform such checks.

The second challenge is to identify the next vertex to visit in the traversal. Consider the case where the subgraph reachable from $s$ is a tree with back edges. Suppose that the current vertex (being explored) has a left child which is not part of a back edge. Then, the pre-order traversal must next visit the left child. In a different case, if the current vertex is the left child of the parent and is also a leaf vertex, then the next vertex to visit in the pre-order is the right child of the parent. However, it could be that the right child is a back edge\footnote{Forward and cross edges are not possible by \cref{lemma:char-of-simple-path}.}. Lange's procedure does not have to handle this case as $\REACHUL$ configuration graphs do not have back edges. 

To understand the third challenge, let us describe how the next vertex in the pre-order traversal is found if there are no more left child or right child left to move. This is because they have already been traversed or because they are back edges or leads to a deeper depth. 
Consider the case when the current vertex is a leaf and is the right child of its parent. Additionally, let the parent be the left child of its grandparent. In this case, the next vertex to visit is the grandparent's right child. However, it could be that this child could form a back edge. 

We handle both of these issues in the $\NEXT$ and $\PATHUNCLE$ procedure (\cref{alg:next}, \cref{alg:path_uncle}). The rest of this paper is organized as follows.

\noindent
\textbf{Organization of the paper:} After giving the necessary background in~\cref{sec:prelims}, we first shows that reachability in simple path graphs is in $\UL \cap  \co \UL$ in \cref{sec:simplepath_mainresult}, which prove \cref{Thm:SPReachULCoUL}. Towards this, we introduce a few procedures such as the $\MAIN$ (explained in \cref{sec:main-procedure}) that performs a depth bounded depth first search on the input graph and checks for reachability. For this traversal, $\MAIN$ uses the procedure $\NEXT$ (explained in \cref{sec:next-procedure}) that returns the pre-order successor of a given vertex in a special subgraph induced on the input graph. Two supporting algorithms used in $\NEXT$ are explained in \cref{sec:back-edge-detect,sec:pathuncle}. Finally, we conclude with a discussion on the recognition vs reachability for special family of graphs in \cref{sec:rec_vs_reach}, proving \cref{Thm:Mangrove_reach} from \cref{sec:intro}.

\paragraph{AI Disclosure} The Ti\textit{k}Z code for \cref{fig:simplepathgrah,fig:uncle-crop,fig:nextreturn1,fig:nextreturn2,fig:nextreturn3,fig:nextreturn4} were generated with the help of ChatGPT-5 and Claude-4. All other parts of this manuscript did not use any form of generative AI (open or proprietary).

\paragraph{Acknowledgements}
We are grateful to the anonymous reviewers for their comments and suggestions and, in particular, for pointing out to us that recognizing simple path graphs also solves the reachability      problem unambiguously (see Page \pageref{Thm:SPReachULCoUL}). C. S. thanks the Department of Computer Science and Engineering, IIT Palakkad for the extended financial support received.

\section{Preliminaries}\label{sec:prelims}

In this section, we give an overview on the complexity classes and graph theory terms used. The standard definitions for complexity classes appear in Kozen \cite{kozen2006} and Arora-Barak \cite{arora2009}. 

As mentioned in the introduction, all the graphs considered are directed. One such graph class discussed in this work is reach unambiguous graphs. It is defined as follows in \cite{LAN97}.
\begin{definition}[Reach Unambiguous Graph $(G,s)$]
	A directed graph $G$ is a reach unambiguous graph if there exists at most one path from the start vertex to any other vertex in $G$. Moreover, the out-degree of each vertex is either 0 or 2 and each vertex is marked as internal or leaf.

\end{definition}

A strict generalization of Reach unambiguous graphs are Simple Path Graphs (\cref{def:simplepathgraph}) introduced by \cite{KKR08}. It is defined as follows.

\begin{figure}[h]
	\begin{minipage}{0.45\linewidth}
		\begin{definition}[Simple Path Graph] For a directed graph $G$ and a vertex $s$, the pair  $(G,s)$ is a simple path graph if there exists at most one simple path from $s$ to any vertex in the graph. Moreover, the out-degree of each vertex is either $0$ or $2$, and each vertex is marked as internal or leaf. For a fixed embedding of the graph, vertex $v$ with non-zero out-degree, we denote $R(v)$ to be $v$'s right child and $L(v)$ to be $v$'s left child.\label{def:simplepathgraph}
		\end{definition}
	\end{minipage}
	\hfill 
	\begin{minipage}{0.45\linewidth}
		\centering
		\begin{tikzpicture}[
			->, >=stealth,
			every node/.style={fill, circle, inner sep=2pt},
			level 1/.style={sibling distance=2cm, level distance=1.2cm},
			level 2/.style={sibling distance=1cm, level distance=1.2cm}
			]
			\node[
			label=above:$s$, 
			] (top) {}
			child { node (L) {} 
				child { node (L1) {} }
				child { node (L2) {} }
			}
			child { node (R) {} 
				child { node (R2) {} }
			};
			\node[right=3cm of current bounding box.north] (S) {}
			child { node (T2) {} }
			child { node (T3) {} };
			\draw (T3) -- (T2);
			\draw (T2) -- (R);
			\draw (T2) -- (R2);
			\draw (T3) -- (R2);
			\draw[bend left=60] (R) to (top);
			\begin{pgfonlayer}{background}
				\node[
				fill=gray!35,
				dashed,
				rounded corners=3pt,
				inner sep=1pt,
				fit=(S)(T2)(T3)
				] {};
			\end{pgfonlayer}
			
		\end{tikzpicture}
		\caption{A simple path graph $(G,s)$. The shaded region in gray denotes the unreachable vertices from $s$.}
		\label[figure]{fig:simplepathgrah}
	\end{minipage}
\end{figure}
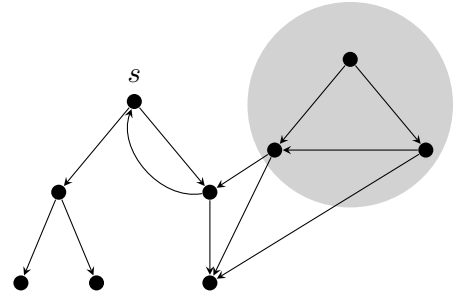
\vspace{2.5mm}

Note that the graph induced by the unreachable vertices from $s$ does not have any restrictions on its structure. Moreover, unlike reach unambiguous graph, simple path graphs may have walks from $s$ to any other vertex reachable in $G$. Additionally, when it is clear from the context, vertex $s$ is omitted while describing a simple path graph.

We also need the definition of $G_d(s)$ (\cref{def:Gds}) and uncle of a vertex in $G_d(s)$ (\cref{def:uncle}) to understand Lange's algorithm. They are defined as follows.

\begin{definition}[$G_d(s)$ of a simple path graph $G$]\label{def:Gds}
	For a simple path graph $(G,s)$, $G_d(s)$ is the induced subgraph of $G$ formed by all vertices reachable from $s$ within distance $d$.
\end{definition}

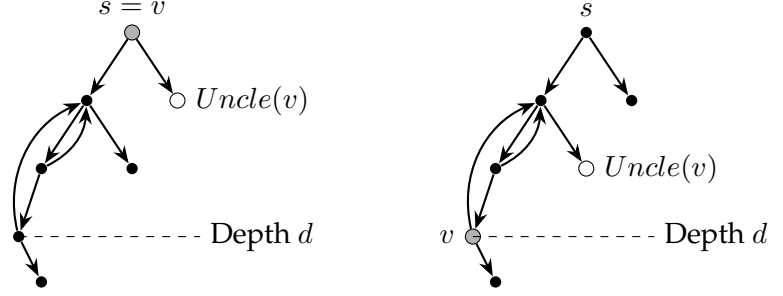
\begin{figure}
	\begin{minipage}{0.35\linewidth}
		\begin{definition}[Uncle of a vertex $v$ for a depth $d$]\label{def:uncle}
			For a simple path graph $(G,s)$, an uncle of a vertex $v$ is defined as follows. If $v$ is $s$, uncle is $R(v)$. Otherwise, uncle is the next vertex to be visited in the pre-order traversal of subtree $G_d(s)$, assuming all the vertices of the subtree of $G_d(s)$ rooted at $v$ is visited. See \cref{fig:uncle-crop} for an example.
		\end{definition}	
		
	\end{minipage}
	\hfill 
	\begin{minipage}{0.62\linewidth}
		\begin{tikzpicture}[scale=0.6,
			node/.style={circle, fill=black, inner sep=1.5pt},
			v_node/.style={circle, draw=black, fill=gray!60, inner sep=2pt},
			uncle_node/.style={circle, draw=black, fill=white, inner sep=2pt},
			edge/.style={-Stealth, thick}
			]
			
			\node[v_node, label={above:$s=v$}] (s) at (2, 3) {};
			\node[node] (parent) at (1, 1.5) {};
			\node[uncle_node, label=right:$Uncle(v)$] (uncle) at (3, 1.5) {};
			\node[node] (v) at (2, 0) {};
			\node[node] (child1) at (0, 0) {};
			\node[node] (below1) at (-0.5, -1.5) {};
			\node[node] (below2) at (0, -2.5) {};
			
			\draw[edge] (s) -- (parent);
			\draw[edge] (s) -- (uncle);
			\draw[edge] (parent) -- (v);
			
			\draw[edge] (parent) to (child1);
			\draw[edge, bend right=30] (child1) to (parent);
			\draw[edge, bend left=40] (below1) to (parent);
			\draw[edge] (child1) -- (below1);
			\draw[edge] (below1) -- (below2);
			
			\draw[dashed] (-0.5, -1.5) -- (3.5, -1.5) node[right, pos=1] {Depth $d$};
			
		\end{tikzpicture}
		\hfill
		\begin{tikzpicture}[scale=0.6,
			node/.style={circle, fill=black, inner sep=1.5pt},
			v_node/.style={circle, draw=black, fill=gray!60, inner sep=2pt},
			uncle_node/.style={circle, draw=black, fill=white, inner sep=2pt},
			edge/.style={-Stealth, thick}
			]
			
			\node[node, label={above:$s$}] (s) at (2, 3) {};
			\node[node] (parent) at (1, 1.5) {};
			\node[node] (uncle) at (3, 1.5) {};
			\node[uncle_node,label=right:$Uncle(v)$] (v) at (2, 0) {};
			\node[node] (child1) at (0, 0) {};
			\node[v_node, label=left:$v$] (below1) at (-0.5, -1.5) {};
			\node[node] (below2) at (0, -2.5) {};
			
			\draw[edge] (s) -- (parent);
			\draw[edge] (s) -- (uncle);
			\draw[edge] (parent) -- (v);
			
			\draw[edge] (parent) to (child1);
			\draw[edge, bend right=30] (child1) to (parent);
			\draw[edge, bend left=40] (below1) to (parent);
			\draw[edge] (child1) -- (below1);
			\draw[edge] (below1) -- (below2);
			
			\draw[dashed] (-0.5, -1.5) -- (3.5, -1.5) node[right, pos=1] {Depth $d$};
			
		\end{tikzpicture}
		\caption{Example illustrating the uncle of a vertex $v$ at different positions.}
		\label{fig:uncle-crop}
	\end{minipage}
\end{figure}

\begin{definition}[DFS Tree, Tree edges]
	A DFS tree $G'$ of a graph $G$ for a DFS traversal (that is, pre-order traversal) is a subgraph of $G$ with $V(G')=V(G)$ and $E(G')$ being the set of edges taken during the traversal. These edges are also known as tree edges.
\end{definition}
In a DFS tree of a directed graph $G$, the edges of $G$ graph are classified as follows. This well known~\cite{CLRS09} formalization is described below.
\begin{definition}[Edge classification in DFS]
	In a depth-first search (DFS) tree of a directed graph, edges in $G$ not appearing in the DFS tree can be classified as follows. A forward edge is an edge that connects a vertex to a descendant in the DFS tree, but not directly to its immediate child. A cross edge is an edge that connects two vertices where neither is a descendant of the other. Finally, a back edge is an edge that points from a vertex to one of its ancestors in the DFS tree. This edge indicates the presence of a directed cycle in a directed graph. The dashed edge in \cref{fig:simplepathgrah} is a back edge.
\end{definition}
Note that trees, by definition, do not contain forward, cross or back edges. Simple path graphs can be viewed as generalizations since they cannot contain forward or cross edges (but may contain back edges). We use this characterization crucially in \cref{sec:simplepath_mainresult}.

\begin{lemma}[Lemma 2 of \cite{KKR08}]\label{lemma:char-of-simple-path}
	A pair $(G,s)$ where $G$ is a directed graph and $s \in V(G)$ is a simple path graph iff for any DFS tree produced from any vertex $v$ reachable from $s$ does not contain forward or cross edges.
\end{lemma}

\noindent
\paragraph*{Complexity Preliminaries and Definitions}
We begin by recalling two unambiguous logspace complexity classes that play a central role in our results. The class $\UL$ captures nondeterministic logspace computations with a unique accepting computation path, while $\RUL$ imposes the stronger requirement of unambiguity, allowing at most one computation path from the start configuration to any configuration. These notions are formalized below.
\begin{definition}[$\UL$]
	$\UL$ is the complexity class containing decision problems that can be solved by a nondeterministic TM using a logarithmic amount of memory space and on every input have at most one computation path from the start configuration to any accepting configuration.
\end{definition}

\begin{definition}[$\RUL$ \cite{BGL91}]
	$\RUL$ is the complexity class containing decision problems that can be solved by a nondeterministic TM using a logarithmic amount of memory space and on every input have at most one computation path from the start configuration to any other configuration.
\end{definition}

Now, we consider the language $L_{ru}$ defined in \cite{LAN97} where it was shown that it is $\RUL$-complete.
\begin{definition}[$L_{ru}$\label{def:lru}]
	The language $L_{ru}$ describes the family of graphs $G$ together with a vertex labeling of $G$ and a source vertex $s$ with the property that from the source node $s$ to every other vertex in $G$, there is at most one path. Moreover, the vertex labeling marks each node as either a leaf or an inner node, with children specified as left or right.
\end{definition}

For a structured family of graphs $\mathcal{G}$, we consider the problem of checking whether a graph belongs to a graph family $\mathcal{G}$ (called $\GREC$) as well as the graph reachability problem restricted to this family of graphs ($\GREACH$) defined as follows.
\begin{definition}
	$\GREC = \{ G \mid G \text{ is a graph and } G \in \mathcal{G} \}$
\end{definition}
\begin{definition}
	$	\GREACH = \{ (G,s,t) \mid G \in \GREC\text{ and there is an } s \text{ to } t \text{ path in} G\}$
\end{definition}

\section{Reachability in Simple Path Graphs is in $\UL \cap \co\UL$}\label{sec:simplepath_mainresult}

In this section, we show that reachability in simple path graphs is in $\UL \cap \co\UL$ (\cref{Thm:SPReachULCoUL}). Before going into the details of the algorithm, we give a brief background for the problem. 

Note that the simple path graphs introduced by \cite{KKR08} is a strict generalisation of the graphs considered in $L_{ru}$ (see~\cref{def:lru}), which is the set of all graphs such that there is at most one path from start to any vertex in the graph. \cite{KKR08} showed that for any $0<\epsilon<1$, recognition as well as the reachability problem in simple path graphs can be done in $\poly(n)$ time and $\tilde{O}(n^{\epsilon})$ space where the $\tilde{O}$ hides the polynomial in $\log(n)$ factors. In this section, we build on an algorithm due to Lange \cite{LAN97}. Lange's algorithm checks whether $G$ belongs to $L_{ru}$, reach unambiguously, and we extend it to our setting. In the next section, we start by discussing how the back edges are detected unambiguously.

\subsection{Detecting Back Edges Unambiguously 
	(\cref{alg:isleftchildbackedge})} \label{sec:back-edge-detect}

In this section, we first describe the procedure $\ISLBACK$ that takes in a vertex $y$ and its height $h$ in the subgraph rooted at $s$ and checks whether $(y, L(y))$ or $(y, R(y))$ is a back edge, respectively. Let $G_h(s)$ denotes this subgraph. The procedure assumes $G_h(s)$ is a simple path graph and the correct height of $y$ is $h$. These checks are not a part of the procedure. The algorithm for $\ISRBACK$ is exactly same as that of $\ISLBACK$ except that in line \ref{alg:isleftchildbackedge:4}, $L(y)$ is replaced by $R(y)$.

\begin{algorithm}[htp!]
	\caption{$\ISLBACK(y,h)$}
	\label{alg:isleftchildbackedge}	
	$flag \gets False$, $w \gets s$\;
	\For{$count \gets 0 $ to $h-1$}{\label[line]{alg:isleftchildbackedge:2}
		
		\If{$w == L(y)$\label{alg:isleftchildbackedge:4}}{
			$flag \gets True$\;
		}
		
		$w \gets$ a neighbor of $w$ guessed nondeterministically\;\label{alg:isleftchildbackedge:6}
		
	}
	\If{$w == y$}{
		Return $flag$\;
	}
	\Else{
		Abort\label{alg:isleftchildbackedge:11}\;
	}
\end{algorithm}

\begin{proposition}\label{prop:backedge}
Given a graph $G$, vertex $y$ and its height $h$ with $G_h(s)$ being a simple path graph, if $(y,L(y))$ is a back edge in the subgraph $G_h(s)$ then $\ISLBACK$ returns True or Aborts. Conversely, if $\ISLBACK$ returns True, then $(y,L(y))$ is a back edge in the subgraph $G_h(s)$. Moreover, $\ISLBACK$ runs unambiguously.
\end{proposition}
\begin{proof}

We start by arguing the correctness of $\ISLBACK$. The argument for the correctness of $\ISRBACK$ is similar.

Suppose $(y,L(y))$ is a back edge in the subgraph $G_h(s)$. Then there exists a path from $s$ to $y$ of length $h$ such that $L(y)$ will be visited before $y$ is reached. The lines \ref{alg:isleftchildbackedge:2}-\ref{alg:isleftchildbackedge:6} precisely does this check and the flag will be set to True in that nondeterministic choice. Hence, the procedure returns True if the path ends in $y$. If the path does not end in $y$, since $h$ is the length of the shortest path from $s$ to $y$, it must be that the nondeterministic choices took at least one back edge during the traversal. In such cases, the procedure aborts in line \ref{alg:isleftchildbackedge:11}. 

We now argue that if $\ISLBACK$ returns True, then $(y,L(y))$ is a back edge in the subgraph $G_h(s)$. To argue this, suppose $(y,L(y))$ is not a back edge. Then, in any path from $s$ to $y$, $L(y)$ will not be visited. Hence, in lines \ref{alg:isleftchildbackedge:2}-\ref{alg:isleftchildbackedge:6}, the flag is never set to True in all nondeterministic choices. Hence, the procedure returns False if the path ends in $y$. If the path does not end in $y$, since $h$ is the length of the shortest path from $s$ to $y$, it must be that the nondeterministic choices took at least one back edge during the traversal. In such cases, the procedure aborts in line \ref{alg:isleftchildbackedge:11}.

\paragraph{Unambiguity of $\ISLBACK$: }We argue the unambiguity of the procedure when it returns True or False. We recall that $G_h(s)$ is already a simple path graph and $h$ is the length of the shortest path from $s$ to $y$. We start with the case when the function returns True. Suppose the $\ISLBACK$ on $(y,h)$ returns True for two nondeterministic choices. This corresponds to two walks from $s$ to $y$. We claim that both of them must be \textit{paths} of length $h$. To see this, suppose one of them is a walk of length $h$ where a vertex $v$ is repeated (such a $v$ must exist). Then we can obtain a path of length strictly smaller than $h$ by removing the vertices that appear between two consecutive $v$. We repeat this for every such $v$. This shows that $y$ is at a distance strictly smaller than $h$, which is a contradiction.  

With two paths of length $h$ from $s$ to $y$, the graph is no longer a simple path graph, which contradicts that $G_h(s)$ is a simple path graph. The argument for the unambiguity of $\ISRBACK$ returning True is similar.

Now, suppose that the procedure returns False. The same argument holds verbatim with ``True" replaced by ``False". Since the procedure returns True and False unambiguously, the procedure is in $\UL \cap \co\UL$.
\end{proof}
\begin{remark}
The above procedures are not reach unambiguous as it is possible to have multiple walks in the graph that reach the same vertex (different from $y$) and end up aborting ambiguously. 
\end{remark}

\subsection{\texorpdfstring{$\MAIN$ Procedure (\cref{alg:main}): Correctness and Unambiguity}{MAIN Procedure: Correctness and Unambiguity}}\label{sec:main-procedure}
In this section, we describe the procedure $\MAIN$. We start with an overview, followed by a detailed description. This procedure takes in a graph $G$, two vertices, $s$ and $t$, and accepts if and only if $G$ is a simple path graph and there is a unique path from $s$ to $t$ in $G$. For each value of $d$, this procedure proceeds by verifying that up to a depth $d$, the given graph $G$ is a simple path graph. 

During this traversal, for each value of the depth, two pointers ($y$ and $y'$) are maintained. The procedure simultaneously checks two properties. First, it checks whether $t$ appears as a neighbor of any visited vertex. Second, for each $y$ at height $d$, it does an inner traversal using a second pointer $y'$ to verify that $G$ does not violate the conditions for a simple path graph for $G_d(s)$. Specifically, the inner traversal checks that $(y,y')$ not a cross edge or a forward edge, and no two internal vertices at depth $d$ share a common neighbor. After we traverse the whole of the graph and all depth values have been processed, the algorithm accepts if and only if $reached$ is True.

For the traversal, this procedure invokes $\NEXT$ that takes in a vertex $y$, its height $h$, and a depth $d$ and returns the next vertex after $y$ in the pre-order traversal of $G_d(s)$. The details of $\NEXT$ as well as the correctness and unambiguity will be explained in the next section (\cref{sec:next-procedure}).

The description of $\MAIN$ is given in~\cref{alg:main}.

\begin{algorithm}[h]
\caption{ $\MAIN(G,s,t)$}\label{alg:main}
$reached \gets$ False, $n \gets |V(G)|$\;
\For{$d \gets 1 \dots n-2$}{
	$(y,h) \gets (s, 0)$\;
	\Repeat{$(y, h) = (s, 0)$}{
		\If{$t \in \{L(y), R(y)\}$}{
			$reached \gets$ True\;}\label{main:6}
		\If{$h == d$ and $y$ is internal}{
			$(y', h') \gets (s, 0)$\;
			\Repeat{$(y', h') = (s, 0)$}{
				\If{$y' \in \{L(y), R(y)\}$}{
					\If{$\neg (\ISLBACK(y,h)$ or $\ISRBACK(y,h))$}{
						
						Reject\;}}\label{main:12}
				\If{$h' == d$ and $y \neq y'$ and $y'$ is internal and $\{L(y'), R(y')\} \cap \{L(y), R(y)\} \neq \emptyset$}{
					Reject\;}\label{main:14}
				$(y',h')=\NEXT(y', h', d)$\;
			}
		}
		$(y,h)=\NEXT(y, h, d)$\;
	}
}
\If{$reached$}{
	Accept\;}
\Else{
	Reject\;}	
	\end{algorithm}

	To argue correctness, we show the following.
	\begin{proposition}\label{prop:main}
Given a graph $G$ and two distinct vertices $s$ and $t$ in $V(G)$, $(G,s,t) \in \SPREACH$ iff $\MAIN$ accepts. Moreover, $\MAIN$ runs unambiguously in Accept and Reject.
\end{proposition}
\begin{proof}

Suppose that $(G,s,t) \not \in \SPREACH$. Then, it could be that either $G$ is not a simple path graph or $G$ is a simple path graph but there is no $s$ to $t$ path. We start with the former case.

Let $p_1$ and $p_2$ be two different simple paths from $s$ to a vertex $v$ in a graph $G$. In the figures that follow (\cref{fig:p1_p2_diff_length,fig:p1_p2_same length}), they are as dashed and dotted respectively. Without loss of generality, let $v$ be the first point of convergence of $p_1$ and $p_2$. Let $y$ appear on $p_1$ as a neighbor of $v$ and $y'$ appear on $p_2$ as a neighbor of $v$. The paths $p_1$ and $p_2$ can either be of the same length or of different length. On one hand, if they are of the same length (see \cref{fig:p1_p2_same length}), either $(y,v)$ or $(y',v)$ forms a cross edge and the graph is rejected in line \ref{main:14}. 

On the other hand, if the paths are of different lengths, without loss of generality, say $y'$ have less height than $y$. Suppose, for a contradiction, the algorithm accepts. Then, the ``Reject'' at line \ref{main:12} was never executed. . We do not consider the ``Reject'' at line \ref{main:14} since this is only executed when $p_1$ and $p_2$ are of same length.

There could be two reasons why the``Reject'' at line \ref{main:12} was never executed. (1) $(y,y')$ is a back edge and $y'$ is a child of $y$,  and (2) $y'$ is not a child of $y$.
For case (1), if $(y,y')$ is a back edge, the path $p_1$ and $p_2$ meet at $y'$ for the first time and not $v$ (see \cref{fig:p1_p2_diff_length}). This is a contradiction as $y'$ is different from $v$. Then the algorithm must have rejected in this case. Now, for case (2), if $y'$ is not a child of $y$, we have to show that the algorithm rejects. Since the inner loop compares $y$ with every $y'$ of height $h' \leq$ height of $y$, for some $y$, $y'$ becomes $v$ and $y$ is the parent of $y'$. Now, $(y,y')$ is not a tree edge (since $y$ and $y'$ appear on different paths) and not a cross edge (which would lead to reject at line \ref{main:12}), it must be a back edge, leading to a ``Reject'' in line \ref{main:12}.

\begin{figure}[htp!]
	\centering
	\begin{minipage}{.45\textwidth}
		
		\includegraphics[width=.85\linewidth]{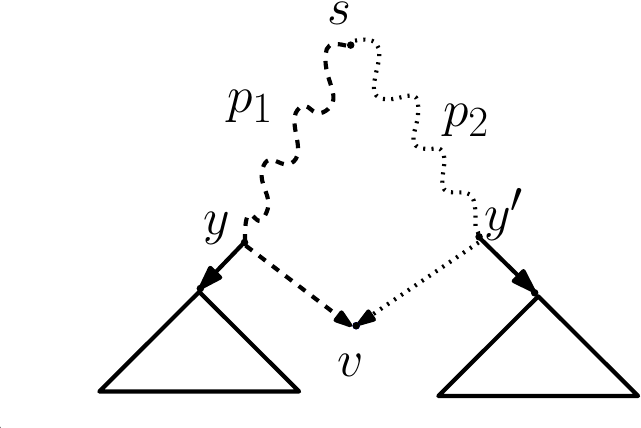}
		\caption{$p_1$ (shown as dashed), $p_2$ (shown as dotted) are of same length}
		\label{fig:p1_p2_same length}
		\vspace{-2em}
	\end{minipage}
	\hspace{1.5em}
	\begin{minipage}{.45\textwidth}
		
		\includegraphics[width=.8\linewidth]{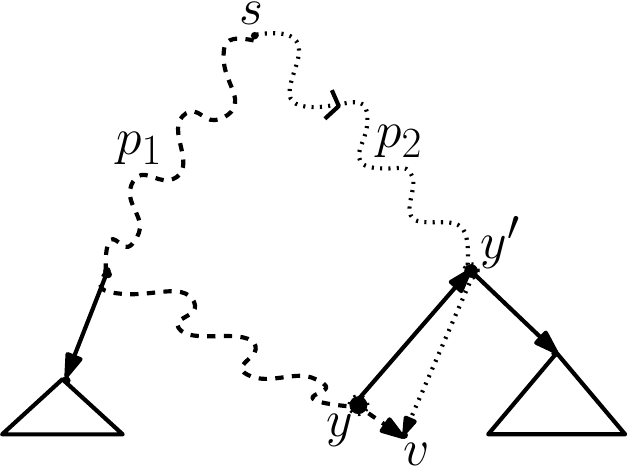}
		\caption{$p_1$ (shown in dashed), $p_2$ (shown in dotted) are of different length.}
		\label{fig:p1_p2_diff_length}
	\end{minipage}
\end{figure}
Now, consider the case where $G$ is a simple path graph, but there is no $s$ to $t$ path. Observe that the algorithm visits every vertex reachable from $s$ in $G$. This is because the $\NEXT$ procedure is guaranteed to find the pre-order successor of the current vertex. Hence, as $t$ was not reachable, the algorithm would reject.

Suppose $(G,s,t) \in \SPREACH$. Then, $G$ is a simple path graph and there is an $s$ to $t$ path. In this case, none of the reject checks succeed, and for some vertex $y$ reachable from $s$, the vertex $t$ will be a child of $y$. The algorithm is guaranteed to find such a $y$ because $\NEXT$ gives the pre-order successor of the current vertex every time it is called. Consequently, $reached$ is set to True in line \ref{main:6}. Hence, the algorithm accepts.

\paragraph{Unambiguity of $\MAIN$: }The main procedure does not use any nondeterminism directly and uses nondeterminism only in the subroutines. Observe that all the subroutines ($\ISLBACK$ and $\ISRBACK$) are unambiguous in both ``yes'' and ``no'' instances by \cref{prop:backedge}.  Moreover, $\NEXT$ returns the pre-order successor of $y$ and its height in the subgraph $G_d(s)$  unambiguously by \cref{prop:next}. Hence, in $\MAIN$, there cannot be two nondeterministic choices that ends in Accept. Same holds for Reject. 
\end{proof}
\begin{proof}[Proof of \cref{Thm:SPReachULCoUL}]
To show that $\SPREACH \in \UL \cap \co\UL$, we see that $\MAIN$ accepts $\SPREACH$ and the accept and reject are unambiguous by \cref{prop:main}. This completes the proof.
\end{proof}
\subsection{\texorpdfstring{$\NEXT$ Procedure (\cref{alg:next}): Correctness and Unambiguity}{NEXT Procedure: Correctness and Unambiguity}} \label{sec:next-procedure}

In this section, we describe the procedure $\NEXT$ as \cref{alg:next}. This procedure takes in a vertex $y$ and two numbers $h$ and $d$ where it is guaranteed that the height of $y$ in the subgraph $G_d(s)$ is $h$. This procedure returns the pre-order successor of $y$ in the subgraph $G_d(s)$. The procedure assumes that $h$ is the correct height of $y$ and $G_d(s)$ is a simple path graph. 

Now, let us briefly see how $\NEXT$ returns a vertex. We consider three cases for a vertex $y$ at height $h<d$. In \cref{fig:nextreturn1,fig:nextreturn2,fig:nextreturn3,fig:nextreturn4}, the vertex highlighted with a gray fill and black border denotes the current vertex $y$, while the vertex with a white fill and black border denotes the vertex returned by $\NEXT(y, h, d)$. The shaded region indicates vertices that are unreachable from $s$. We now describe the cases.
\begin{description}
\item[Case 1:] Left child of $y$ is not a back edge. In this case, $L(y)$ is returned as the pre-order successor (\cref{fig:nextreturn1}). 

\item[Case 2:] Left child $y$ is a back edge, but the right child is not. In this case, $R(y)$ is returned as the pre-order successor (\cref{fig:nextreturn2}).

\item[Case 3:] Both children of $y$ are back edges. Then, pre-order successor of $y$ in $G_d(s)$ is returned (\cref{fig:nextreturn3}).

\item[Case 4:] Now, suppose $y$ is at height $h=d$. Even if $y$ has children that are at a greater depth than $d$, $\NEXT$ does not return it. Instead $\NEXT$ returns the pre-order successor of $y$ in $G_d(s)$ (\cref{fig:nextreturn4}).

\end{description}
The algorithm is described in \cref{alg:next}. Following this, we argue the correctness.

\begin{figure}
\centering

\begin{minipage}{0.49\linewidth}
	\centering
	\begin{tikzpicture}[
		node/.style={circle, fill=black, inner sep=1.2pt},
		y_node/.style={circle, draw=black, fill=gray!60, inner sep=2pt},
		empty_node/.style={circle, draw=black, fill=white, inner sep=2pt},
		edge/.style={-Stealth, thick}
		]
		
		\filldraw[fill=gray!40, draw=gray] (3.5, 2.1) circle (0.8cm);
		
		\node[node] (b1) at (3.5, 2.5) {};
		\node[node] (b2) at (3, 1.8) {};
		\node[node] (b3) at (4, 1.8) {};
		\draw[edge] (b1) -- (b2);
		\draw[edge] (b2) -- (b3);
		\draw[edge] (b3) -- (b1);
		
		\node[node, label=above:$s$] (s) at (1.2, 2.5) {};
		\node[y_node, label=left:$y$] (y) at (0.5, 1.5) {};
		\node[empty_node] (next) at (0.5, 0.5) {};
		\node[node] (mid) at (2.2, 1.5) {};
		\node[node] (depth_node) at (2.2, 0.5) {};
		
		\draw[edge] (s) to (y);
		\draw[edge, bend right=30] (y) to (s);
		\draw[edge] (s) to (mid);
		\draw[edge, bend left=30] (mid) to (s);
		\draw[edge] (y) -- (next);
		\draw[edge,bend right=20] (b1) to (mid);
		\draw[edge] (b2) to (mid);
		\draw[edge] (mid) -- (depth_node);
		\draw[edge, bend left=20] (b3) to (depth_node);
		
		\draw[dashed] ([xshift=-0.5cm]next.center) to node[below, pos=1, xshift=0cm] {Depth $d$} ([xshift=0.8cm]depth_node.center);
		\node[anchor=north] at (next.south) {$\NEXT(y, h, d)$};
		
	\end{tikzpicture}
	\caption{Case 1: When left child of $y$ is not a back edge.}
	\label{fig:nextreturn1}
\end{minipage}
\hfill
\begin{minipage}{0.49\linewidth}
	
	\begin{tikzpicture}[
		node/.style={circle, fill=black, inner sep=1.2pt},
		y_node/.style={circle, draw=black, fill=gray!60, inner sep=2pt},
		empty_node/.style={circle, draw=black, fill=white, inner sep=2pt},
		edge/.style={-Stealth, thick}
		]
		
		\filldraw[fill=gray!40, draw=gray] (3.5, 2.1) circle (0.8cm);
		
		\node[node] (b1) at (3.5, 2.5) {};
		\node[node] (b2) at (3, 1.8) {};
		\node[node] (b3) at (4, 1.8) {};
		\draw[edge] (b1) -- (b2);
		\draw[edge] (b2) -- (b3);
		\draw[edge] (b3) -- (b1);
		
		\node[node, label=above:$s$] (s) at (1.2, 2.5) {};
		\node[node] (y) at (0.5, 1.5) {};
		\node[node] (next) at (0.5, 0.5) {};
		\node[y_node, label=below left:$y$] (mid) at (2.2, 1.5) {};
		\node[empty_node] (depth_node) at (2.2, 0.5) {};
		
		\draw[edge] (s) to (y);
		\draw[edge, bend right=30] (y) to (s);
		\draw[edge] (s) to (mid);
		\draw[edge, bend left=30] (mid) to (s);
		\draw[edge] (y) -- (next);
		\draw[edge,bend right=20] (b1) to (mid);
		\draw[edge] (b2) to (mid);
		\draw[edge] (mid) -- (depth_node);
		\draw[edge, bend left=20] (b3) to (depth_node);
		
		\draw[dashed] ([xshift=-0.5cm]next.center) to node[below, pos=0, xshift=-0.3cm] {Depth $d$} ([xshift=0.8cm]depth_node.center);
		\node[anchor=north] at (depth_node.south) {$\NEXT(y, h, d)$};
		
	\end{tikzpicture}

	\caption{Case 2: When left child of $y$ is a back edge, but the right child is not.}
	\label{fig:nextreturn2}
\end{minipage}

\end{figure}
\begin{figure}
\centering
\begin{minipage}{0.49\linewidth}
	\centering
	\begin{tikzpicture}[
		node/.style={circle, fill=black, inner sep=1.2pt},
		y_node/.style={circle, draw=black, fill=gray!60, inner sep=2pt},
		empty_node/.style={circle, draw=black, fill=white, inner sep=2pt},
		edge/.style={-Stealth, thick}
		]
		
		\filldraw[fill=gray!40, draw=gray] (3.5, 2.1) circle (0.8cm);
		
		\node[node] (b1) at (3.5, 2.5) {};
		\node[node] (b2) at (3, 1.8) {};
		\node[node] (b3) at (4, 1.8) {};
		\draw[edge] (b1) -- (b2);
		\draw[edge] (b2) -- (b3);
		\draw[edge] (b3) -- (b1);
		
		\node[node, label=above:$s$] (s) at (1.2, 2.5) {};
		\node[y_node, label=left:$y$] (y) at (0.5, 1.5) {};
		\node[empty_node, label={below=2mm, xshift=2mm:$\NEXT(y, h, d)$}] (mid) at (2.2, 1.5) {};
		\node[node] (depth_node) at (2.2, 0.5) {};
		
		\draw[edge] (s) to (y);
		\draw[edge, bend right=30] (y) to (s);
		\draw[edge] (s) to (mid);
		\draw[edge, bend left=30] (mid) to (s);
		\draw[edge, bend left=30] (y) to (s);
		\draw[edge,bend right=20] (b1) to (mid);
		\draw[edge] (b2) to (mid);
		\draw[edge] (mid) -- (depth_node);
		\draw[edge, bend left=20] (b3) to (depth_node);
		
		\draw[dashed] ([xshift=-0.5cm]next.center) to node[below, pos=1, xshift=0cm] {Depth $d$} ([xshift=0.8cm]depth_node.center);
		
	\end{tikzpicture}
	\caption{Case 3: When both children of $y$ are back edges.}
	\label{fig:nextreturn3}
\end{minipage}
\hfill
\begin{minipage}{0.49\linewidth}
	\centering
	
	\begin{tikzpicture}[
		node/.style={circle, fill=black, inner sep=1.2pt},
		y_node/.style={circle, draw=black, fill=gray!60, inner sep=2pt},
		empty_node/.style={circle, draw=black, fill=white, inner sep=2pt},
		edge/.style={-Stealth, thick}
		]
		
		\filldraw[fill=gray!40, draw=gray] (3.5, 2.1) circle (0.8cm);
		
		\node[node] (b1) at (3.5, 2.5) {};
		\node[node] (b2) at (3, 1.8) {};
		\node[node] (b3) at (4, 1.8) {};
		\draw[edge] (b1) -- (b2);
		\draw[edge] (b2) -- (b3);
		\draw[edge] (b3) -- (b1);
		
		\node[node, label=above:$s$] (s) at (1.2, 2.5) {};
		\node[node] (y) at (0.5, 1.5) {};
		\node[y_node,label=above left:$y$] (next) at (0.5, 0.5) {};
		\node[empty_node, label={below=2mm, xshift=2mm:$\NEXT(y, h, d)$}] (mid) at (2.2, 1.5) {};
		\node[node] (depth_node) at (2.2, 0.5) {};
		
		\draw[edge] (s) to (y);
		\draw[edge, bend right=30] (y) to (s);
		\draw[edge] (s) to (mid);
		\draw[edge, bend left=30] (mid) to (s);
		\draw[edge] (y) -- (next);
		\draw[edge,bend right=20] (b1) to (mid);
		\draw[edge] (b2) to (mid);
		\draw[edge] (mid) -- (depth_node);
		\draw[edge, bend left=20] (b3) to (depth_node);
		
		\draw[dashed] ([xshift=-0.5cm]next.center) to node[below, pos=1, xshift=0cm] {Depth $d$} ([xshift=0.8cm]depth_node.center);
		
	\end{tikzpicture}
	\caption{$y$ is at height $h=d$.}
	\label{fig:nextreturn4}
\end{minipage}

\end{figure}

\begin{figure}[ht!]
\begin{minipage}{0.53\columnwidth}
	\begin{algorithm}[H]
		\caption{ $\NEXT(y,h,d)$}\label{alg:next}	
		\If{$y$ is internal and $\neg\ISLBACK(y,h)$ and $h < d$}{
			Return $(L(y), h+1)$\;
		}
		\Else{
			\If{$y$ is internal and $\neg\ISRBACK(y,h)$ and $h < d$}{
				Return $(R(y), h+1)$\;
			}
			\Else{
				$h' \gets 0$\;
				$w \gets s$\;
				$flag \gets$ False\;
				\While{$w$ is internal and $h' \leq d$ and $flag == $\textrm{False}}{\label{alg:next:11}
					\If{$\neg\ISRBACK(w,h)$}{
						$w \gets R(w)$\;
						$h' \gets h' + 1$\;
					}
					\ElseIf{$\neg\ISLBACK(w,h)$}{
						$w \gets L(w)$\;
						$h' \gets h' + 1$\;
					}
					\Else{
						$flag \gets$ True\;
					}
				}\label{alg:next:22}
				\If{$y == w$}{
					Return $(s, 0)$\;
				}
				\Else{
					$(u,h_u) = \PATHUNCLE(y,h)$\;
					
					Return $(u,h_u)$
				}
			}
		}
	\end{algorithm}
	
\end{minipage}
\begin{minipage}{0.5\columnwidth}
	\begin{algorithm}[H]
		\caption{ $\PATHUNCLE(x,g)$}\label{alg:path_uncle}
		
		$u \gets R(s)$, $w \gets s$\;
		$count \gets 1$, $h_u \gets 1$, $h_w \gets 0$\;
		\While{$count \leq g$ and $w$ is an internal vertex}{ \label{alg:path_uncle:3}
			Guess $b \in \{0,1\}$\;
			\If{$b == 0$}{
				\If{$\neg\ISLBACK(w,h_w)$}{
					
					\If{$\neg\ISRBACK(w,h_w)$}{\label{alg:path_uncle:22}
						$u \gets R(w)$\;\label{alg:path_uncle:14}
						$h_u \gets h_w+1$\;
					}
					$w \gets L(w)$\;
					$h_w \gets h_w+1$\;
					
				}
			}
			\Else{
				\If{$\neg\ISRBACK(w,h_w)$}{
					$w \gets R(w)$\;
					$h_w\gets h_w+1$\;
				}
			}
			$count\gets count+1$\label{alg:path_uncle:16}\;
		}
		\If{$w==x$}{Return $(u,h_u)$\; }
		\Else{Abort\;\label{alg:path_uncle:19}}
	\end{algorithm}
	
\end{minipage}
\caption{Algorithms $\NEXT$ and $\PATHUNCLE$}
\end{figure}

\begin{proposition}\label{prop:next}
Given a graph $G$, vertex $y$, its height $h$ and depth $d$ of the subgraph $G_d(s)$, $\NEXT$ returns the pre-order successor of $y$ in the graph $G_d(s)$ and its height. Moreover, $\NEXT$ returns $y$ and its height unambiguously or aborts.
\end{proposition}
\begin{proof}

The purpose of the procedure $\NEXT(y,h,d)$ is to return the pre-order successor of $y$ in the graph $G_d(s)$. Below, we show that this indeed is the case. The pre-order successor of $y$ can appear as a left child of $y$ if the left child is not a back edge. If it is a back edge, the pre-order successor will be the right child of $y$, provided $(y,R(y))$ is not a back edge. If both left and right child of $y$ are back edges, then the pre-order successor is not immediately clear as it needs to be computed. This is achieved using a call to $\PATHUNCLE$.

Suppose the next vertex in the pre-order traversal is the left child and $(y,L(y))$ is not a back edge. In this case, the procedure returns the left child and terminates. Suppose the next vertex in the pre-order traversal is the right child and $(y,R(y))$ is not a back edge. In this case, the procedure returns the right child and terminates. 

If the current vertex is the final vertex in the pre-order traversal of the graph $G_d(s)$, we return the source and terminate the procedure. Computing this final vertex is achieved in lines \ref{alg:next:11} to \ref{alg:next:22}. If $y$ is not the final vertex, we call $\PATHUNCLE$ to return the next vertex.

\paragraph{Unambiguity of $\NEXT$: }The procedure $\NEXT$ is unambiguous since it contains no nondeterministic statements except within the subroutines $\PATHUNCLE$, $\ISLBACK$, and $\ISRBACK$. The subroutine $\ISLBACK$ and $\ISRBACK$ are unambiguous in both ``yes'' and ``no'' instances by \cref{prop:backedge}. Path uncle unambiguously returns the correct uncle of $x$ along with its height by \cref{prop:path_uncle}.  Hence, $\NEXT$ is unambiguously returns a value or aborts.
\end{proof}

We remark that the only difference with the $\NEXT$ procedure appearing in Lange \cite{LAN97} is that in our case, (1) we need to check for back edges (\cref{lemma:char-of-simple-path}) and (2) the ``right most" vertex in $G_d(s)$ need not be the last vertex in the pre-order traversal; this needs to be computed.

\subsection{\texorpdfstring{$\PATHUNCLE$ Procedure (\cref{alg:path_uncle}): Correctness and Unambiguity}{PATHUNCLE: Correctness and Unambiguity}}\label{sec:pathuncle}

This procedure is executed as a subroutine of the previously defined algorithm $\NEXT$. Given a vertex $x$ and its height $g$, the algorithm $\PATHUNCLE(x, g)$ unambiguously returns 
the correct uncle of $x$ along with the uncle’s height. This is achieved by guessing a path from the root $s$ to $x$, at each step storing the latest encountered potential uncle $u$ (i.e., the right child of the current node) and its corresponding height $h_u$. To ensure guesses form a valid path, the algorithm avoids traversing edges that are identified as back edges, using the procedures $\ISLBACK$ and $\ISRBACK$. The algorithm verifies the correctness of its guesses by checking that the guessed path ends at $x$. If the path is invalid, the algorithm aborts. Otherwise, it returns the uncle that is farthest from $s$ along the simple path to $x$.

\begin{proposition}\label{prop:path_uncle}
	Given a vertex $x$ and its height $g$, the algorithm $\PATHUNCLE(x, g)$ unambiguously returns 
	the correct uncle of $x$ along with the uncle’s height or aborts.
\end{proposition}
\begin{proof}

	We argue the correctness of $\PATHUNCLE$. 
	
	Note that there is exactly one path $P_{s,x}$ from $s$ to $x$ as $x$ is reachable from $s$ and $G_d(s)$ is a simple path graph.
	Observe that the while loop (lines \ref{alg:path_uncle:3} to \ref{alg:path_uncle:16}) eventually terminates because, in each iteration, the value of $count$ is incremented irrespective of the guess taken.
	
	We show that the uncle $u$ must appears as a right child of one of the vertices in $P_{s,x}$.

	To argue this, we claim the following ``(1) Uncle of $w$ is $u$ with $h_u$ as its height (2) $h_w$ is the height of $w$" is maintained at the end of line \ref{alg:path_uncle:16}.  Item (2) is necessary for ensuring correctness of call to $\ISLBACK$ and $\ISRBACK$ functions.

	We argue the invariant by induction on $count$. For $count=1$, $w=s$ and the uncle $u$ of $w$ is $R(s)$. Moreover $h_u = 1$ and $h_w=0$ are the correct heights of $u$ and $w$ respectively. Hence, the base case holds.
	
	Moving to the induction case, let $count=c$ for some $c \ge 1$ and $w$ be an internal vertex reached. By induction hypothesis, the uncle of $w$ is $u$ and $h_u$ is its height and $h_w$ is the height of $w$. We make the following claim.

	\begin{proposition}\label{claim:structuncle}
		Let $G$ be a simple path graph and $w$ be a vertex such that $(w,R(w))$ is a back edge in $G$. Then, the uncle of $w$ is the right child of the unique vertex $v$ in the unique path $P_{s,w}$ such that 
		
		\begin{enumerate}
			\item[a.] $v$ has a right child that is not a back edge
			\item[b.] Every vertex $r$ from $v$ to $w$ in $P_{s,w}$ the edge $(r,R(r))$ is either an edge in $P_{s,w}$ or is a back edge. 
		\end{enumerate} 
	\end{proposition}

	In the induction step, $w$ gets updated to  $R(w)$ or  $L(w$) based on the nondeterministic choice. We now argue that for both the cases, $h_w$ gets updated correctly. This would prove (2) from the induction. In both the cases, provided $L(w)$ or $R(w)$ is \textit{not} a back edge, $w$ is updated appropriately, and $h_w$ is incremented by $1$. By induction, $h_w$ had the correct height of $w$. Hence, incrementing $h_w$ gives the correct height of the new $w$. 
	
	We now argue (1) by induction. Suppose that the nondeterministic choice was to the right (that is, $b=1$) and the right child is verified to be not a back edge. In this case, $u$ and $h_u$ stay the same. Hence, the uncle of $R(w)$ is the same as that of the uncle of $w$, which is $u$. 
	
	Suppose that the nondeterministic choice was to the left (that is, $b=0$) and the left child is verified to be not a back edge. Now, there are two cases. \textit{Case 1}: The right child of $w$ is a back edge. In this case, the uncle and its height are not updated. With $(w,R(w))$ being a back edge, from \cref{claim:structuncle}, we conclude that the  uncle of $L(w)$ is the same as that of uncle of $w$. By induction hypothesis, with the uncle of $w$ being $u$, the uncle of $L(w)$ is also $u$ and height of $u$ is the same as it is not updated. \textit{Case 2}: Right child of $w$ is not a back edge. In this case, in step \ref{alg:path_uncle:14}, the uncle is updated to $R(w)$. This is correct since the conditions of \cref{claim:structuncle} are satisfied. Since $u$ is updated to be the right child of $w$, the height of $u$ is one more than the height of $w$.
	
	Hence, by induction, irrespective of the nondeterministic choice, the uncle of $w$ is always $u$. For the unique nondeterministic choice reaching $x$, $u$ holds the uncle of $x$. In the case where the verification steps (back edge checks) fail, the $count$ gets incremented and becomes equal to $g$ prematurely. Since $g$ is the length of the shortest path to $x$, $w$ never becomes $x$, and the procedure aborts.

	\paragraph{Unambiguity of $\PATHUNCLE$: }Now, we show that the algorithm is unambiguous. Towards this, we argue that for the execution of the procedure in which all nondeterministic choices are correct, the procedure returns the correct uncle and its associated height. On the other hand, for any execution path involving an incorrect nondeterministic choice, the procedure aborts.
	
	We first argue the case when all the nondeterministic choices taken by the algorithm are correct. We show a unique $(u,h_u)$ is returned. This is true since for every $w$, there is a unique path $P_{s,w}$ from $s$ to $w$ in $G_d(s)$ and there is a unique vertex $v$ in the path $P_{s,w}$ whose right child is the uncle. Hence, the returned uncle is unique.
	
	Now consider the case in which the procedure makes an incorrect nondeterministic guess, which can occur in one of two ways: (1) the path of length $g$ chosen by the procedure starting from vertex $s$ does not terminate at vertex $x$, or (2) the path length is less than $g$ because $w$ reached a leaf vertex in $G_d(s)$. In both cases, the path terminates at vertex $w$, different from $x$ as $g$ is the length of the shortest path from $s$. Here, the procedure aborts in line \ref{alg:path_uncle:19}. These aborts need not be unique as there could be multiple nondeterministic choices that could end in the same $w$.
\end{proof}
We now give a proof of \cref{claim:structuncle} which was used in proving \cref{prop:path_uncle}.
\begin{proof}[Proof of \cref{claim:structuncle}]Observe that such a $v$ satisfying (a) and (b) always exists. This is because the right child of $s$ is a candidate uncle for any vertex in $P_{s,w}$. The uniqueness follows since the pre-order traversal is unique.
	
	Suppose, for contradiction, the right child of $v$ is not the correct uncle of $w$ but satisfies conditions (a) and (b). Observe that the uncle of $w$ cannot appear as a right child of any vertex in the path $P_{s,v}$. Hence, the uncle of $w$ must appear as a right child of a vertex $r$ between $v$ and $w$ in $P_{s,w}$. However, for any $r$, $(r, R(r))$ is neither an edge in $P_{s,w}$ nor a back edge. This violates (b).  
\end{proof}

\section{Recognition vs Reachability for Special Family of Graphs}\label{sec:rec_vs_reach}

Having established that reachability in simple path graphs is in $\UL \cap \co\UL$, we now turn to a broader question about the complexity of recognition versus reachability in various special family of graphs.

Let $\mathcal{G}$ be a family of graphs. Recall the recognition ($\GREC$) and reachability ($\GREACH$) problems from preliminaries.
Clearly, the reachability problem requires the recognition problem to be solved as well. This means $\GREC \leq_m^{log} \GREACH$, showing that solving reachability is as ``hard'' as recognition.
Since $\GREACH$ where $\mathcal{G}$ is reach unambiguous graphs is in  $\RUL$ \cite{LAN97}, the recognition problem is also solvable in $\RUL$ by the aforementioned observation. We now show that the reachability problem also reduces to the recognition problem for the family of reach unambiguous graphs.

We now prove \cref{Thm:Mangrove_reach}, which shows $\GREACH\leq_m^{\log} \GREC$ for reach unambiguous graphs.

\begin{proof}[Proof of \cref{Thm:Mangrove_reach}]

	The reach unambiguous machine $M'$ with oracle access to $L(M)$ for $\GREACH$ is as follows. Given $(G,s,t)$, first check if $G \in \GREC$ by making an oracle query to $L(M)$. If $G \notin \GREC$, then reject. Otherwise, $G \in \GREC$, and guess a length $\ell < \mid V(G)\mid $. Now, guess a path of length $\ell$ from $s$ to $t$. If the path ends in $x$ where $x\ne t$, abort with $(\ell, x )$. On the other hand, if $x=t$, then accept.

	We show that $L(M')=\GREACH$. Suppose $M'$ accepts $(G,s,t)$, then it must be that $G \in \GREC$ and there is a path from $s$ to $t$ in $G$. Hence, $(G,s,t) \in \GREACH$.  Suppose $(G,s,t) \in \GREACH$, then $G \in \GREC$ and there is a path from $s$ to $t$ in $G$. Hence the oracle call to $L(M)$ succeeds. Moreover, there is exactly one guess for length $\ell< |V(G)|$  and exactly one path starting from $s$ that ends in $t$. Hence, the machine $M'$ accepts $(G,s,t)$.
	
	We now argue why $M'$ is reach unambiguous.
	
	Consider the case where $\mathcal{G}$ is reach unambiguous graphs. Suppose $(G,s,t) \notin \GREACH$. Either $G \notin \GREC$ or there is no $s$ to $t$ path. Note that the case where there is more than one path from $s$ to $t$ is covered in the $G \notin \GREC$. Suppose $G \notin \GREC$. Then the oracle will return a ``no" answer and $M'$ rejects unambiguously. Now, suppose $G \in \GREC$ and there is no path from $s$ to $t$ in $G$. Then $(\ell,x)$ must be unique for all nondeterministic choices made by $M'$. Hence, $M'$ is reach unambiguous.
	
	Since $\RUL^\RUL=\RUL$ by \cref{Thm:RUL^RUL}, $L(M') = \GREACH $ is in $\RUL$. 
\end{proof}
\section{Summary and Open Problem}

The reachability problem restricted to reach unambiguous graphs is complete for the class $\RUL$ (under many-one logspace reductions)~\cite{LAN97}. In this work, we studied the complexity of reachability problems specifically for simple path graphs (a superclass of reach unambiguous graphs) and showed that the problem is in $\UL\cap\co\UL$. Since $\RUL \subseteq \UL \cap \co\UL$, as an open question, we ask whether our result can be improved to $\RUL$.

\medskip
\noindent\textbf{Open Problem.}
Is the complexity of the reachability problem for simple path graphs, $\SPREACH \in \RUL$?
\medskip

Such an improvement would give a simultaneous $\poly(n)$ time and $O(\log^2 n )$ space algorithm for this problem, placing it in $\SC^2$ and thereby improving \cite{KKR08}.

\bibliographystyle{alpha}
\bibliography{../../../ref}

\begin{thebibliography}{CLRS09}

\bibitem[AB09]{arora2009}
S.~Arora and B.~Barak.
\newblock {\em Computational Complexity: A Modern Approach}.
\newblock Cambridge University Press, 2009.

\bibitem[AE25]{AE24}
Simon Apers and Roman Edenhofer.
\newblock Directed st-connectivity with few paths is in quantum logspace.
\newblock In {\em Proceedings of the 40th Computational Complexity Conference},
  CCC '25. Schloss Dagstuhl--Leibniz-Zentrum fuer Informatik, 2025.

\bibitem[{\`A}G00]{ACR00}
Carme {\`A}lvarez and Raymond Greenlaw.
\newblock A compendium of problems complete for symmetric logarithmic space.
\newblock {\em Computational Complexity}, 9:123--145, 12 2000.

\bibitem[AL98]{AL99}
Eric Allender and K-J Lange.
\newblock $\mathsf{RUSPACE} (\log n)\subseteq \mathsf{DSPACE} (\log^2 n/\log
  \log n)$.
\newblock {\em Theory of Computing Systems}, 31(5):539--550, 1998.

\bibitem[All23]{E23}
Eric Allender.
\newblock Guest column: Parting thoughts and parting shots (read on for details
  on how to win valuable prizes!).
\newblock {\em SIGACT News}, 54(1):63–81, March 2023.

\bibitem[BJLR91]{BGL91}
Gerhard Buntrock, Birgit Jenner, Klaus-J\"{o}rn Lange, and Peter Rossmanith.
\newblock Unambiguity and fewness for logarithmic space.
\newblock In {\em Proceedings of the 8th International Symposium on
  Fundamentals of Computation Theory}, FCT '91, page 168–179, Berlin,
  Heidelberg, 1991. Springer-Verlag.

\bibitem[BTV09]{BTV09}
Chris Bourke, Raghunath Tewari, and N.~V. Vinodchandran.
\newblock Directed planar reachability is in unambiguous log-space.
\newblock {\em ACM Trans. Comput. Theory}, 1(1), February 2009.

\bibitem[CLRS09]{CLRS09}
T.H. Cormen, C.E. Leiserson, R.L. Rivest, and C.~Stein.
\newblock {\em Introduction to Algorithms, third edition}.
\newblock Computer science. MIT Press, 2009.

\bibitem[Coo79]{CookS79}
Stephen~A. Cook.
\newblock Deterministic cfl's are accepted simultaneously in polynomial time
  and log squared space.
\newblock In {\em Proceedings of the Eleventh Annual ACM Symposium on Theory of
  Computing}, STOC '79, page 338–345. Association for Computing Machinery,
  1979.

\bibitem[GST19]{GST19}
Chetan Gupta, Vimal~Raj Sharma, and Raghunath Tewari.
\newblock {Reachability in O(log n) Genus Graphs is in Unambiguous Logspace}.
\newblock In Rolf Niedermeier and Christophe Paul, editors, {\em 36th
  International Symposium on Theoretical Aspects of Computer Science (STACS
  2019)}, volume 126, pages 34:1--34:13, 2019.

\bibitem[KKR08]{KKR08}
Sampath Kannan, Sanjeev Khanna, and Sudeepa Roy.
\newblock {STCON in Directed Unique-Path Graphs}.
\newblock In Ramesh Hariharan, Madhavan Mukund, and V~Vinay, editors, {\em
  IARCS Annual Conference on Foundations of Software Technology and Theoretical
  Computer Science}, volume~2, pages 256--267, Dagstuhl, Germany, 2008.

\bibitem[Koz06]{kozen2006}
Dexter Kozen.
\newblock {\em Theory of Computation}.
\newblock Texts in Computer Science. Springer, 2006.

\bibitem[Lan97]{LAN97}
Klaus-J{\"o}rn Lange.
\newblock An unambiguous class possessing a complete set.
\newblock In {\em Annual Symposium on Theoretical Aspects of Computer Science},
  pages 339--350. Springer, 1997.

\bibitem[RA00]{RA2000}
Klaus Reinhardt and Eric Allender.
\newblock Making nondeterminism unambiguous.
\newblock {\em SIAM Journal on Computing}, 29(4):1118--1131, 2000.

\bibitem[Rei08]{Rein08}
Omer Reingold.
\newblock Undirected connectivity in log-space.
\newblock {\em Journal of the ACM (JACM)}, 55(4):1--24, 2008.

\bibitem[Wig92]{Wig}
Avi Wigderson.
\newblock The complexity of graph connectivity.
\newblock In {\em International Symposium on Mathematical Foundations of
  Computer Science}, pages 112--132. Springer, 1992.

\end{thebibliography}

\end{document}